\pdfoutput=1 
\pdfoutput=1 
\pdfoutput=1 
\pdfoutput=1 
\pdfoutput=1
\documentclass[letterpaper, 10 pt, conference]{ieeeconf} % Comment this line out
\IEEEoverridecommandlockouts                              % This command is only
\usepackage{graphics} % for pdf, bitmapped graphics files
\usepackage{epsfig} % for postscript graphics files
\usepackage{times} % assumes new font selection scheme installed
\usepackage{amsmath} % assumes amsmath package installed
\usepackage{amsmath} 
\DeclareMathOperator*{\argmin}{arg\,min}
\usepackage{amssymb}
\usepackage{color}
\usepackage[colorlinks=false, urlcolor=blue, pdfborder={0 0 0}]{hyperref}
\usepackage{enumerate}
\usepackage[linesnumbered,boxed,ruled,commentsnumbered]{algorithm2e}
\usepackage{subfigure}
\usepackage{comment}
\usepackage{cite}
\usepackage{flushend}
\title{\LARGE \bf
Secure-by-Construction Optimal  Path Planning for \\
Linear Temporal Logic Tasks}

\author{Shuo Yang, Xiang Yin, Shaoyuan Li and Majid Zamani
	\thanks{This work was  supported by the National Natural Science Foundation of China (61803259, 61833012) and by Shanghai Jiao Tong University Scientific and Technological Innovation Funds.
		}
\thanks{Shuo Yang, Xiang Yin and Shaoyuan Li are with Department of Automation
and Key Laboratory of System Control and Information Processing,
Shanghai Jiao Tong University, Shanghai 200240, China.
        {\tt\small  \{xiang-yang,yinxiang,syli\}@sjtu.edu.cn.}}%
\thanks{M. Zamani is with the Computer Science Department, University of Colorado Boulder, CO 80309, USA.  
	M. Zamani is also with the Computer Science Department, Ludwig Maximilian University of Munich, 80539 Munich, Germany.	
	Email: {\tt\small majid.zamani@colorado.edu}.     
}}

\newtheorem{mydef}{Definition}
\newtheorem{mythm}{Theorem}
\newtheorem{myprob}{Problem}

\newtheorem{mypro}{Proposition}

\newtheorem{remark}{Remark}

\IncMargin{0.8em}
\begin{document}

\maketitle

%%%%%%%%%%%%%%%%%%%%%%%%%%%%%%%%%%%%%%%%%%%%%%%%%%%%%%%%%%%%%%%%%%%%%%%%%%%%%%%%
\begin{abstract}
In this paper, we investigate the problem of planning an optimal infinite path for a single robot to achieve a linear temporal logic (LTL) task with security guarantee. 
We assume that the external behavior of the robot, specified by an output function, can be accessed by a passive intruder (eavesdropper).
The security constraint requires that the intruder should never infer that the robot was started from a secret location. 
We provide a sound and complete algorithmic procedure to solve this problem.  
Our approach is based on the construction of the twin weighted transition systems  (twin-WTS) that tracks a pair of paths having the same observation. 
We show that the security-aware path planning problem can be effectively solved based on graph search techniques in the product of the twin-WTS and the B\"{u}chi automaton representing the LTL formula.  
The complexity of the proposed planning algorithm is polynomial in the size of the system model.
Finally, we illustrate our  algorithm by a simple robot planning example. 
\end{abstract}

\section{Introduction}\label{sec:1}
\subsection{Motivation}
Path planning is a fundamental problem in robotics  which asks to generate a planned trajectory from an initial location such that some desired requirements are fulfilled. 
Classical planning problems usually focus on low-level tasks such as obstacle avoidance or point-to-point navigation \cite{lavalle2006planning,karaman2011sampling}. 
In the past decades, temporal-logic-based high-level path planning for complex tasks has drawn considerable attention in the literatures; see, e.g., \cite{fainekos2009temporal,wongpiromsarn2012receding,kress2018synthesis,kloetzer2019path}. 
In this framework, the planning task is specified by linear temporal logic (LTL) or computation tree logic (CTL) formulae.  
In particular,  LTL can be  used to represent many important properties such as safety,  liveness  and priority \cite{baier2008principles}. 
By using automata-theoretic approach, algorithmic procedures are developed to automatically generate correct-by-construction plans  to achieve the given temporal logic tasks.

While the temporal-logic-based planning has  been extensively investigated for safety requirements,  security and privacy requirments are left as an afterthought in many applications.  For instance, in  robot data collecting problem, a robot needs to visit different locations in order to  gather data and then to transmit collected data  to the cloud.  However, the data transmission    may not be secure in the sense that there may exist an eavesdropper ``listening" to the communication. Such information leakage may reveal some crucial secret behavior of the robot, e.g., some information, the robot does not intend to transmit,  may be \emph{inferred} by the intruder. Therefore, one also needs to incorporate such a security constraint in the path planning algorithm. 
Due to its importance, security and privacy concerns have been attracting  attentions in the robot path planning literature; see, e.g., \cite{li2019coordinated,zheng2020adversarial}.

\subsection{Our Contributions} 
Motivated by the security concerns in robotic systems, in this paper, we formulate and solve a security-aware optimal path planning problem with respect to LTL requirments. Specifically, we consider a single robot whose mobility is modeled as a weighted transition system (WTS). We consider an intruder modeled as an outside observer (or eavesdropper) who accesses the external behaviors of the system specified by an output function. 
We consider the planning problem of achieving a task specified by a general LTL formula, while hiding the secret initial location of the robot. 
To capture this security requirement, we adopt the notion of an information-flow security property called \emph{initial-state opacity}  \cite{saboori2013verification}. 
Specifically, a planed path from a secret initial-state is said to be \emph{secure} if there exists another path from a non-secret initial-state such that those two paths are observationally equivalent from the intruder's point of view. 

Our approach is different from the standard initial-state opacity verification procedure \cite{saboori2013verification}, which requires to build the initial-state estimator whose size is exponential in the number of system states. Instead, we propose a computationally more efficient approach by constructing the twin-WTS structure which synchronizes the system with its copy based on the observation. Similar structures have been used in the literature for the purpose of property verification, e.g., diagnosability, observability and prognosability.
Here, we show that the security-aware path planning problem can be effectively solved by a graph search in the product of the twin-WTS and the B\"{u}chi automaton that accepts the given LTL task. 
Also, we show that the constructed product system also preserves optimality.  
Furthermore, our algorithm fails to provide a solution only when no solution exists.
Hence, we provide a sound and complete solution to the security-aware optimal LTL planning problem. 

\subsection{Related Works}
Optimal LTL path planning problem was originally formulated in \cite{smith2011optimal}, where  the optimization objective is to minimize the worst cost between each satisfying instances. 
This framework has been extended to the case of multi-robot \cite{guo2015multi,ulusoy2013optimality}, where each robot may have a local task or a team of robots need to collaborate to achieve a global task. 
Recently, sampling-based techniques have been applied to improve the scalability of optimal path planning algorithm \cite{li2017sampling,kantaros2018sampling}. 
Optimal temporal logic path planning problems  have also been studied for  stochastic systems modeled as MDPs; see, e.g., \cite{wolff2012robust,ding2014optimal,deng2017approximate,guo2018probabilistic}. 
However, none of the above mentioned works considers security constraints.

In the context of security-aware path planning, our work is mostly related to \cite{hadjicostis2018trajectory}. 
The differences between our work and \cite{hadjicostis2018trajectory} are as follows. 
First, the planning task considered in our work is expressed as a general LTL formula, while  \cite{hadjicostis2018trajectory} considers a simple reachability task. 
Second, no optimality is consider in \cite{hadjicostis2018trajectory}. 
Finally, the security requirement considered in our work is different with that in \cite{hadjicostis2018trajectory}. 
In particular, \cite{hadjicostis2018trajectory} considers protecting the \emph{current} secret location of the robot, 
while we consider protecting the \emph{initial} secret location of the robot. 
We show that initial-type secret has a nice property and it suffices to track a pair of observational equivalent states in the system. 
Therefore, the complexity of our planning algorithm is \emph{independent} from the number of secret states and is always quadratic in the number of system states. 
However, the  complexity of the planning algorithm  in \cite{hadjicostis2018trajectory} is based on the  structure of $K$-detector, whose size grows exponentially in the number of secret states.

In the computer science literature, the concept of \emph{hyper-properties} \cite{clarkson2010hyperproperties} has attracted many attentions in the past years, e.g., HyperLTL \cite{clarkson2014temporal}. In particular, hyper-properties are closely related to security requirements as it allows to specify the relationships among multiple paths. 
Very recently,  the authors of  \cite{wang2020hyperproperties} show that initial-state opacity planning problem can be specified as an instant of the HyperLTL planning problem; symbolic algorithms for finite synthesis is also provided therein. 
This result is closely related to ours. 
However, initial-state opacity considered in  \cite{wang2020hyperproperties} is based on the equivalence of atomic propositions. In our setting, atomic propositions are only used to specify the desired temporal logic task, while the observation equivalence is specified by a new output function. This setting is more general as the atomic propositions and the output sets can be different.  
Furthermore, our planning algorithm is tailored to initial-state security, 
which avoids the general large complexity in  HyperLTL synthesis.  

Finally, our work is also related to opacity-enforcing supervisory control in the context of discrete-event systems \cite{dubreil2010supervisory,saboori2011opacity,yin2016uniform,tong2018current}. 
However, the opacity-enforcing control problem is essentially a reactive synthesis problem under security constraint whose complexity is exponential in the size of the system.  Here, we consider a security-aware path planning problem that can be solved more efficiently. Furthermore, no LTL specification and optimality were considered in the opacity control problem in \cite{dubreil2010supervisory,saboori2011opacity,yin2016uniform,tong2018current}.

\subsection{Organization}
The rest of the paper is organized as follows. 
In Sections~\ref{sec:2} and~\ref{sec:3},  a motivating example and
some necessary preliminaries are presented, respectively. 
The security-aware LTL planning problem is formally formulated in  Section~\ref{sec:4}. In Section~\ref{sec:5}, we discuss how to solve this problem based on the twin-WTS. A case study is presented in Section~\ref{sec:6} to illustrate the proposed algorithm. Finally, we conclude the paper by Section~\ref{sec:7}.

\section{Motivating Example}\label{sec:2}

\begin{figure}
	\center
	\includegraphics[width=0.45\textwidth]{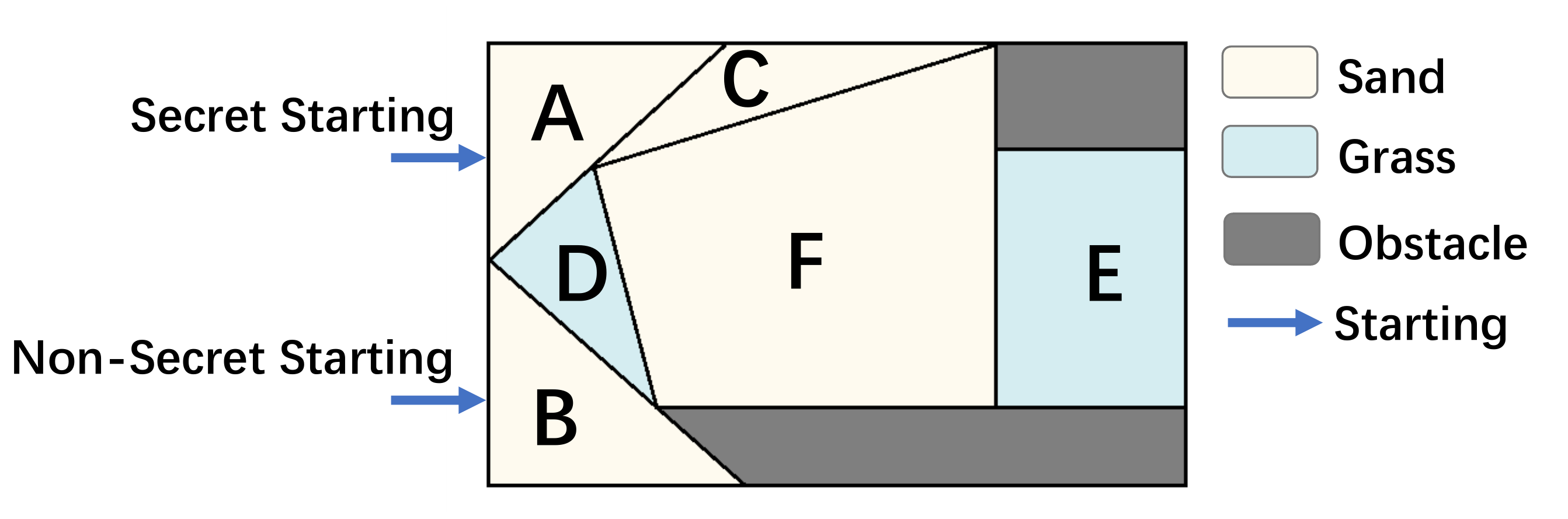}
	\caption{Work space of the single robot.}\label{fig:situation}
\end{figure}

\begin{figure}
	\center
	\includegraphics[width=0.38\textwidth]{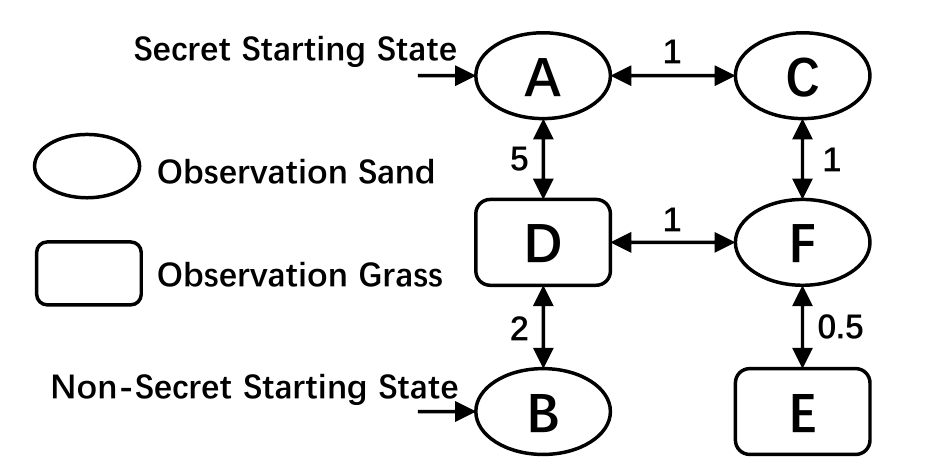}
	\caption{The specification automaton of the motivating example. The intruder has two observations on the robot: the robot is at sand or grass land. The robot can start from A (secret) or B (non-secret). Bidirectional transition means that robot can move in both directions; numbers beside the transition represent the cost of this  transition.}\label{fig:case}
\end{figure}
Before we formally formulate the main problem in this paper, we first consider a motivating example. 
Suppose that a single mobile robot moves in a workspace with grass and sand lands as shown in Figure~\ref{fig:situation}. 
The workspace is partitioned to six regions of interest and  black regions denote obstacles. At each instant, the robot can only move to regions that are adjacent to its current region (sharing at least an edge). 
We assume that the robot always knows exactly its current location. 
On the other hand, we assume that there is an \emph{outside observer} that knows whether the robot is currently at a grass or a sand land. 
The mobility model of the robot can be represented by the transition system shown in  Figure~\ref{fig:case}. 
Furthermore, we assume that there is a cost moving from one region to another, which is specified by the number associated to each bidirectional transition in Figure~\ref{fig:case}.

The task of the robot is to deliver goods between regions $F$ (representing, e.g.,  a factory) and $E$ (representing, e.g.,  a warehouse), i.e, visit $F$ and $E$ infinitely often. 
The robot may initially start from regions $A$ or $B$. 
However, it does not want the outside observer to know that it started from region $A$ (if so). This may because, for example, starting from different locations implies that different robot-types are used, which may further reveal which kind of goods the factory is delivering.

Now, suppose that the robot is starting from region $A$.  
Clearly, the optimal plan to achieve the temporal logic task  is
\[
A\to C \to (  F \to E )^\omega,
\]
where notation $\omega$  over parentheses means the infinite repetition of the finite execution inside them.
However, this plan is not secure in the sense that the observer will know for sure that the robot started from region $A$ after observing two consecutive $Grass$. 
This is because there is no feasible path from region $B$ that can generate the same observation. 
On the other hand, the robot may take the plan below 
\[
A\to D \to (  F \to E )^\omega.
\]
This plan is more expensive as the robot will incur higher cost when moving from $A$ to $D$. 
However, this plan is secure in the sense that there exists another path 
\[
B\to D \to (  F \to E )^\omega, 
\]
starting from region $B$ that generates the same observation. 
Therefore, although more cost is paid, the robot is able to hide the secret about its initial location.

\section{Temporal Logic Task Planning}\label{sec:3}
In this section, we define  basic notations that we use in the rest of the paper and  introduce some necessary preliminaries. 
For a set $A$, we denote by $|A|$ and $2^{A}$  its cardinality and its power set, respectively. 
A finite sequence over $A$ is a sequence in the form of $a_1\cdots a_n$, where $a_i\in A$; 
we denote by $A^*$  the set of all finite sequences over $A$. 
Similarly, we denote by $A^\omega$  the set of all infinite sequences over $A$.

\subsection{Weighted Transition Systems}

We consider a scenario where single mobile robot works in a  workspace $\mathcal{W}\subseteq \mathbb{R}^2$. 
The workspace is partitioned as $n$ disjoint regions of interest denoted by $r_1,\dots,r_n$ and we denote by $\mathcal{I}:=\{1,\cdots,n\}$ the index set. 
In general, workspace regions can be of any arbitrary shape partitioned based on the task properties and the dynamic of the robot; see, e.g., \cite{kloetzer2007temporal,belta2017formal} for details on region partition. 
In this work, we focus on the task planning problem; hence, 
we model the mobility of the robot in the workspace as  a \emph{Weighted Transition System} (WTS) defined as follows.
\begin{mydef}\label{def1}(Weighted Transition System) 
A weighted transition system is a 6-tuple 
\[
T=(Q,Q_0, \to, w,\mathcal{AP},L),  
\]
where
\begin{itemize}
  \item 
   $Q=\{q_i : i\in\mathcal{I}\}$ is the set of states and each state $q_{i}$ indicates that the robot is at location $r_i$;
  \item 
  $Q_0 \subseteq Q$ is the set of initial states representing all possible starting locations of the robot;
  \item 
  $\to \subseteq Q   \times Q $ is the transition relation such that $(q_i,q_j)\in \to$ means that there exists a controller that can drive robot from region $r_i$ to $r_j$  without going through any other regions;
  \item 
  $w: Q\times Q\to \mathbb{R}_{+} $ is a cost function that assigns to each transition $(q_i,q_j)\in \to$ a positive weight $w(q_i,q_j)$ representing the cost incurred for driving the robot from region $r_i$ to $r_j$, e.g., the distance between $r_i$ and $r_j$; 
  \item 
  $\mathcal{AP}$ is the set of atomic propositions used for representing some basic properties of interest;
  \item 
  $L: Q  \to 2^\mathcal{AP}$ is the labeling function that assigns each state to a set of atomic propositions.
\end{itemize}  
\end{mydef}

Given a WTS $T$, an \emph{infinite internal path} is an infinite sequence of states  
                $\tau=\tau(1)\tau(2)\tau(3)\cdots \in Q^\omega$ 
such that
$\tau(1)\in Q_0$  and $(\tau(i), \tau(i+1))\in \to, \forall i\in \mathbb{N}_+$. 
A \emph{finite internal path} of a WTS is defined analogously.  
Hereafter, an internal path will just be referred to as path for the sake of simplicity. 
We denote by $\texttt{Path}^{\omega}(T)$ and $\texttt{Path}^{*}(T)$ the set of all infinite and finite paths in $T$, respectively.
The cost function $w$ is considered to be additive; therefore, 
the cost of a finite path  $\tau\in \texttt{Path}^{*}(T)$, denoted by $J(\tau)$, is defined by
\begin{equation}\label{eq1}
	J(\tau) = \sum_{i=1}^{|\tau|-1}w(\tau(i),\tau(i+1)),
\end{equation}
where $|\tau|$ is the length of the path. In words, the cost $J(\tau)$ captures the total cost incurred during the execution of  finite path $\tau$.
The \emph{trace} of an infinite path $\tau \in Q^\omega $  denoted by $\texttt{trace}(\tau)$ is an infinite sequence over $2^\mathcal{AP}$ such that 
$\texttt{trace}(\tau)=L(\tau(1))L(\tau(2))L(\tau(3))\cdots$.  
Given a set of states $Q'\subseteq Q$, we denote by $\texttt{Reach}(Q')$ the set of states reachable from $Q'$. 
We say  a state $q\in Q$ is in a cycle of $T$ if there exists a sequence $q_1q_2\dots q_k\in Q^*$ such that 
$q_1=q_k=q$ and 
$(q_i,q_{i+1})\in \to ,\forall i\in \mathbb{N}_+$. 
We denote by $\texttt{cycle}(T)$ the set of all states that are in some cycles of $T$.

\subsection{Linear Temporal Logic and B\"{u}chi Automata}

Let $\mathcal{AP}$ be the set of atomic propositions.
A Linear Temporal Logic (LTL) formula is constructed based on atomic propositions, Boolean, and temporal operators. 
Specifically, an LTL formula $\phi$ is recursively defined by
\[
\phi 
::= 
true  
\mid p 
\mid \phi_1\wedge \phi_2
\mid \neg \varphi
\mid \bigcirc \phi
\mid \phi_1 U \phi_2,
\]  
where $p\in \mathcal{AP}$ is an atomic proposition;  $\bigcirc$ and $U$ denote, respectively, ``next" and ``until". 
The above syntax also induces   temporal operators 
$\Diamond$ (``eventually") and $\Box$ (``always"), 
where $\Diamond \phi:=true U \phi$ and $\Box \phi:= \neg \Diamond \neg \phi$.  

LTL formulas are used to evaluate whether or not \emph{infinite words} satisfy some properties. 
Formally, an infinite word $\sigma\in (2^\mathcal{AP})^\omega$ is an infinite sequence over alphabet $2^\mathcal{AP}$.
We denote by $\sigma \models \phi$ if $\sigma$ satisfies the LTL formula $\phi$.  
For example, $\Box \Diamond \phi$ means that property $\phi$ should be satisfied infinitely often.
The reader is referred to \cite{baier2008principles} for more details about the syntax and the semantics of LTL, which are omitted here for the sake of brevity.  
We define $\texttt{Words}(\phi) = \{ \sigma \in (2^\mathcal{AP})^\omega : \sigma \models \phi\}$ as the set of all words satisfying LTL formula $\phi$.

\begin{mydef}(Nondeterministic B\"{u}chi Automaton)
A Nondeterministic B\"{u}chi Automaton (NBA) is a 5-tuple 
$B=(Q_B,Q_{0,B},\Sigma, \to_B, F_B)$, 
where  $Q_B$ is the set of states, $Q_{0,B} \subseteq Q_B$ is the set of initial states, 
$\Sigma$ is an alphabet, $\to_B\subseteq Q_B \times \Sigma  \times Q_B$ is the transition relation
and  $F_B \subseteq Q_B$ is the set of accepting  states.
\end{mydef}

Given an infinite word $\sigma=\pi_0\pi_1\pi_2\cdots \in \Sigma^\omega$, 
an infinite \emph{run} of $B$ over $\sigma$ is an infinite sequence 
$\rho=q_0q_1q_2\cdots\in Q_B^\omega$ such that
$q_0\in Q_{0,B}$ and  $(q_i,\pi_i,q_{i+i})\in \to_B$ for any $i\in \mathbb{N}$. 
An infinite run $\rho\in Q_B^\omega$ is said to be \emph{accepted} by $B$ if $\texttt{Inf}(\rho)\cap F_B\neq \emptyset$, 
where $\texttt{Inf}(\rho)$ denotes the set of states that appears infinite number of times in $\rho$. 
Then, an infinite word $\sigma$ is said to be accepted by $B$ if it induces an infinite run accepted by $B$.  
We denote by $\mathcal{L}_B\subseteq \Sigma^\omega$ the set of all accepted words by NBA $B$.

For any LTL formula $\phi$, it is well-known  \cite{vardi1986automata}  that
there always exists an NBA over $\Sigma=2^{\mathcal{AP}}$ that accepts exactly all infinite words satisfying $\phi$, i.e., $\mathcal{L}_B= \texttt{Words}(\phi)$.
Throughout this paper, $B=(Q_B,Q_{0,B},2^{\mathcal{AP}} , \to_B, F_B)$ is used to denote the NBA corresponding to the LTL formula $\phi$ of interest.

\subsection{Temporal Logic Path Planning}

The standard LTL path planning problem asks to find an infinite path  $\tau\in \texttt{Path}^\omega(T)$ of system $T$ such that $\texttt{trace}(\tau) \models \phi$. Due to the structure of the accepting condition in B\"{u}chi automata, it suffices to find an infinite path with the following \emph{prefix-suffix structure}
\[
\tau=q_1\cdots q_{k}[q_{k+1}\cdots q_{k+m} ]^\omega \in \texttt{Path}^\omega(T)
\]
such that $\texttt{trace}(\tau) \in \mathcal{L}_B$. 
Intuitively,  
$q_{k+1}\cdots q_{k+m}$ is the suffix that forms a cycle such that the robot should execute infinitely often, while $q_1\cdots q_{k}$ is the prefix representing the  transient path that leads to the cyclic path. Such a prefix-suffix structure is also referred to a \emph{plan}.
In this work, we consider the cost of a plan, which is an infinite path, as the cost of its prefix and suffix, i.e., 
\begin{equation}\label{eq2}
\hat{J}(\tau)=J(q_1\cdots q_kq_{k+1}\cdots q_{k+m}q_{k+1}).
\end{equation} 
In order to find an optimal plan with the least cost, 
one can perform modified shortest path search in the \emph{product system} composed by  $T$ and $B$; see, e.g., \cite{ulusoy2013optimality}.  

\begin{remark}
The cost function defined in \eqref{eq2} essentially treats the transient cost $J_{pre}=J(q_1\cdots q_kq_{k+1})$ and the steady-state cost $J_{suf}=J(q_{k+1}\dots q_{k+m}q_{k+1})$ equivalently. 
In general, we can define the cost function as $\hat{J}(\tau)=\alpha  J_{pre}+(1-\alpha)J_{suf}$, 
where $\alpha\in [0,1]$ is a parameter adjusting the weight of each part. 
Our work  considers the case of $\alpha=0.5$ for the sake of simplicity; all results can be easily extended to the general case. 
\end{remark}

\begin{remark}
Given an infinite path (plan), depending on how we decompose   prefix and suffix, the plan may have different costs.
For example, $q_1(q_2q_3)^\omega$ and $q_1q_2(q_3q_2)^\omega$ are the same path but have different costs. 
Hereafter, for a plan $\tau$, $\hat{J}(\tau)$ is always considered as the cost for the  prefix-suffix structure of $\tau$ having the minimum cost. 
\end{remark}	
	
\section{Security-Aware Path Planning Problem}\label{sec:4}
As we discussed in the motivating example, the solution to the standard LTL path planning problem does not necessarily provide security guarantees. In this section, we present the considered information-flow security model and formulate the security-aware path planning problem.

Given WTS $T=(Q,Q_0, \to, w,\mathcal{AP}, L)$, we assume that the internal state of the system is not available to the intruder (malicious observer) directly. 
Instead, the intruder can only infer the behavior of the system via its outputs. 
Formally, we model the intruder's observation of the system  as an output function 
\[
H: Q\to Y,
\] 
where $Y$ is the set of outputs. 
The execution of any  infinite internal path             
$\tau=\tau(1)\tau(2)\tau(3)  \cdots \in  \texttt{Path}^{\omega}(T)$ 
will generate an infinite \emph{external path} 
$H(\tau(1))H(\tau(2))H(\tau(3))\cdots\in Y^\omega$; 
we also denote this external path by $H(\tau)$ with a slight abuse of notation. 
A finite external path is defined analogously.

In this work, we consider the problem of protecting   \emph{secret initial location} of the robot.
To this end, we assume that $Q_S\subset Q_0$ is the set of \emph{secret initial states}. 
Hereafter, a WTS $T$ equipped with output function $H$ and secret initial states $Q_S$ is also written as 
$T=(Q,Q_0, \to, w,\mathcal{AP}, L, H,Y,Q_S)$ for simplicity. 
To guarantee security, we want to make sure that the intruder is not able to infer confidentially that the robot started from a secret location. 
This requirement is formalized  as follows.
 
\begin{mydef}\label{def:security}(Security)
Let $T=(Q,Q_0, \to,w,\mathcal{AP}, L, $ $ H,Y,Q_S)$ be a WTS. 
An infinite path $\tau\in \texttt{Path}^\omega(T)$ is said to be \emph{secure}
if there exists an infinite path $\tau'\in \texttt{Path}^\omega(T)$ such that $\tau'(1)\notin Q_S$ and $H(\tau)=H(\tau')$.
\end{mydef}

\begin{remark}
	The above definition of security is related to the notion of {initial-state opacity} proposed in \cite{saboori2013verification}. Essentially, initial-state opacity is a system property such that all paths generated by the system are secure in our sense. However, as we are considering path planning problem, security is defined only for a specific path rather than the entire system. 
\end{remark}

\begin{myprob}(Security-Aware Optimal LTL Path Planning Problem)\label{prob1}
Given a WTS $T$, secret states $Q_S\subset Q$, output function $H:Q\to Y$ and LTL formula $\phi$, for each possible initial-state $q_0\in Q_0$, 
determine a plan $\tau\in \texttt{Path}^\omega(T)$ with $\tau(1)=q_0$  such that the following conditions hold:
\begin{enumerate}
	\item 
	$\texttt{trace}(\tau)\models \phi$; 
	\item 
	$\tau$ is secure; 
	\item
	For any other plan $\tilde{\tau}\in \texttt{Path}^\omega(T)$ satisfying the above requirements, we have 
	$\hat{J}(\tau)\leq \hat{J}(\tilde{\tau})$.
\end{enumerate}
\end{myprob}

\begin{remark}(Intruder Model) In the above problem formulation, it essentially assumes that the intruder knows the followings: 
\begin{enumerate}
 \item 
 the mobility model of the robot, i.e.,  WTS   $T$; and 
 \item 
 the external path generated by the robot, i.e., $H(\tau)$. 
\end{enumerate}
However, it does not know the exact internal state of the robot, which has to be inferred by observing outputs. 
On the other hand, the robot is assumed to know exactly its initial and current state; therefore, this is still a planning problem under perfect information from the robot's point of view.  
This setting is reasonable in many applications because: 
 (i) the system usually has more ability to acquire information about itself than the intruder; and 
(ii) the intruder's information sometimes comes from eavesdropping the information transmission  which is  a partial information of the robot's knowledge.
\end{remark}

\begin{remark}	
According to Definition~\ref{def:security}, if a path $\tau$ is started from a non-secret initial state $q_0\in Q_0\setminus Q_S$, then it is always secure as we can choose $\tau'=\tau$. 
Therefore, for non-secret initial states, we just need to solve the standard optimal LTL path planning problem; see, e.g., \cite{smith2011optimal}. 
However, for those secret initial states, the security constraint has to be taken into account. 
This issue will be addressed in the next section.
\end{remark}

\section{Planning Algorithm}\label{sec:5}
In this section, we present the security-aware path planning algorithm. 
Our approach is based on constructing a new transition system that effectively captures the security constraint. 
\subsection{Twin-WTS}
In order to handle the security constraint, one needs to track the information of the outside observer based on the external path. Such an information-tracking task can be achieved by constructing the \emph{initial-state estimator} \cite{saboori2013verification}. However, the size of an initial-state estimator grows exponentially as the number of states in the system increases due to the subset construction. 

Here, we present a computationally more efficient approach that does not rely on the construction of the initial-state estimator. 
Instead, we propose a new structure called the \emph{twin-WTS}, which is used to track all current states pairs of two paths that have the same external path from the intruder's point view.
This structure is formally defined as follows.

\begin{mydef}(Twin-WTS) 
Given a WTS $T=(Q,Q_0, \to,$ $w,\mathcal{AP}, L, $ $ H,Y,Q_S)$,  its twin-WTS is a new WTS
\[
V=(X,X_0,\to_{V}, w_V,\mathcal{AP},L_V), 
\]
where
\begin{itemize}
	\item 
	$X\subseteq Q\times Q$ is the set of states; 
	\item 
	$X_0=\{(q_1,q_2)\in Q_0\times Q_0:  H(q_1)=H(q_2)  \}$
	is the set of initial-states; 
	\item 
	$\to_{V} \subseteq X\times X$ is the transition relation defined by: 
	for any $x=( q_1,q_2 )\in X$ 
	    and $x'=( q_1',q_2' )\in X$, 
	  we have $(x,x')\in \to _V$
	if  the followings hold:
	\begin{itemize} 
		\item 
		$(q_1, q_1')\in \to$; 
		\item
		$(q_2, q_2')\in \to$;  
		\item 
		$H(q_1')=H(q_2')$.
	\end{itemize} 
   \item 
   $w_V: X \times X \to \mathbb{R}_{+} $ is the cost function defined by: for any  $x = (q_1,q_2) \in X$ and $x' = (q_1',q_2') \in X$, 
   we have $w_V (x,x')=w(q_1,q_1')$; 
   \item 
    $L_V: X \to 2^{\mathcal{AP}} $ is the labeling function defined by: for any  $x = (q_1,q_2) \in X$, 
     we have $L_V (x)=L(q_1)$. 
\end{itemize}
\end{mydef}
 
\begin{remark}
Intuitively, the twin-WTS  tracks  two internal paths that generate the same external path. 
Specifically, the first component is used to represent the trajectory in the real system, while the second component is used to represent a copy that mimics the real system in the sense of output equivalence. 
Therefore, for any path $(\tau_1(1),\tau_2(1))(\tau_1(2),\tau_2(2))\cdots$ in $V$, we have 
$H(\tau_1(1))H(\tau_1(2))\cdots= H(\tau_2(1))H(\tau_2(2))\cdots$. 
On the other hand, for any two paths $\tau_1, \tau_2$ in $T$ such that $H(\tau_1)=H(\tau_2)$, 
we can find a path $\tau$ in $V$ such that its first component is $\tau_1$ and the second component is $\tau_2$. 
Also, we note that the cost function $w_V$ and the labeling function $L_V$ are all defined based on the states in the first component, which is the part for the real system. 
Finally,  the size of $V$ is polynomial in the size of $T$ as it contains at most $|Q|^{2}$ states.
\end{remark}

\subsection{Planning Algorithm}

The twin-WTS can be used to capture the security constraint based on the following observation. 
For any secure path starting from a secret initial state $q_{s,0}$, there must exist an observation-equivalent path from a non-secret initial state $q_{ns,0}$. 
Furthermore, such a path-pair should exist in the twin-WTS $V$ from state $(q_{s,0}, q_{ns,0})$. 
Therefore, to perform security-aware path planing, it suffices to perform   planning from an initial-state in $V$ in which the first component is the real (secret)  initial-state and the second component is a non-secret state. 
Furthermore, in order to incorporate the temporal task, we need to synchronize the 
twin-WTS with the NBA $B$ that accepts $\phi$; this is defined as the product system.
 
\begin{mydef}(Product System)
Given twin-WTS $V=(X,X_0, \to_V, w_V,\mathcal{AP},L)$ 
and        NBA $B=(Q_B,Q_{0,B},\Sigma,$ $\to_B,F_B)$, 
the product of $V$ and $B$ is a new (unlabeled) WTS
\[
T_{\otimes}=(Q_{\otimes},Q_{0,\otimes},\to_{\otimes},w_{\otimes}), 
\]
where
\begin{itemize}
\item 
$Q_{\otimes}\subseteq    X\times Q_B$ is the set of states;
\item 
$Q_{0,\otimes}=    X_0\times Q_{0,B}$  is the set of initial states; 
\item 
$\to_{\otimes}\subseteq Q_{\otimes}\times Q_{\otimes}$ is the transition relation defined by: 
for any  $q_\otimes = ( x, q_B ) \in Q_\otimes$ and $q_\otimes' = ( x', q_B' ) \in Q_\otimes$, 
we have $(q_\otimes,q_\otimes')\in \to_\otimes$ if the followings hold: 
	\begin{itemize} 
	\item
	$(x,x')\in \to_V$; and   
	\item 
	$(q_{B}, L_V(x)  , q'_{B}) \in \to _B$.
   \end{itemize}
\item 
$w_\otimes: Q_\otimes \times Q_\otimes \to \mathbb{R}_{+} $ is the cost function defined by: for any  $q_\otimes = (x,q_{B}) \in Q_\otimes$ and $q'_\otimes = (x',q'_{B}) \in Q_\otimes$, 
we have $w_\otimes (q_\otimes,q_\otimes')=w(x,x')$.
\end{itemize}
\end{mydef}
 
Essentially, the product system further restricts the dynamic of $V$ such that each movement should satisfy the LTL task $\phi$, i.e., $(q_{B}, L_V(x)  , q'_{B}) \in \to _B$.  
Note that the original WTS $T$ is not synchronized with $B$ as the dynamic of $T$ has already been encoded in the first component of $V$. 
For each state $((q, q'),q_B)\in Q_\otimes$, we denote  by   $\Pi[  (  (q, q'),q_B ) ] =q$ the projection to the state space of $T$;  
  we also write 
$\Pi[  ( (q_0, q'_0), q_{0,B}) \cdots ( (q_n, q'_n), q_{n,B} ) ]=q_0\cdots q_n$.

For each initial-state $q_0\in Q_0$ in $T$, we denote by $\textsc{Int}_{q_0}(T_\otimes) \subseteq Q_{0,\otimes}$ the set of initial-states in $T_{\otimes}$ whose first components are $q_0$ while the second component are non-secret states in $T$, i.e., 
\[
\textsc{Int}_{q_0}(T_\otimes)=\{ (  (q_0,q_0'), q_B ) \in Q_{0,\otimes}:  q_0'\notin  Q_{S}    \}. 
\]
Also, we define $\textsc{Goal}(T_\otimes)\subseteq Q_\otimes$ as the set of  states in $T_{\otimes}$ whose last components are in $F_B$ and they are in some cycles of $T_\otimes$, i.e., 
\begin{align}
&\textsc{Goal}(T_\otimes)  =\nonumber\\
&\{ (  (q,q'), q_B ) \!\in\! Q_\otimes:   q_B\!\in\! F_{B}\wedge 
(  (q,q'), q_B )\!\in\! \texttt{cycle}(T_\otimes)
     \}. \nonumber
\end{align} 
In order to find an optimal path from initial state $q_0$ in $T$, it suffices to find an optimal path in the form of
\[
\textsc{Int}_{q_0}(T_\otimes)\to  (\textsc{Goal}(T_\otimes) \to \textsc{Goal}(T_\otimes) )^\omega
\]
in $T_\otimes$. Note that both sets $\textsc{Int}_{q_0}(T_\otimes)$ and $\textsc{Goal}(T_\otimes)$ are non-singleton in general. 
Therefore, we need to consider all possible combinations in order to determine an optimal path. 
This idea is formalized by Algorithm~1. 

Specifically,  lines 1-3   construct   the NBA $B$,  the twin-WTS $V$ and the product system $T_{\otimes}$.
Line 4 aims to determine if there is a feasible path from $q_0$ satisfying both the LTL constraint and the security constraint. 
In particular, if $\textsc{Int}_{q_0}(T_\otimes)$ cannot reach any goal state in cycle, then this means that there does not exist an infinite path accepted by $\phi$ that has an observation equivalent path from a non-secret initial state, i.e.,  there exists no feasible path starting from $q_0$. 
Otherwise, 
we consider, in lines~7 and~8, each combination of state   $q_I$ in $\textsc{Int}_{q_0}(T_\otimes)$ and  state $q_{G}$ in $\texttt{Reach}(\{q_I\}) \!\cap \! \textsc{Goal}(T_\otimes)$, which is a goal state reachable from $q_1$ and in some cycles. 
In lines 9-10, 
we determine the shortest path from $q_I$ to $q_G$ and the shortest path from $q_G$ back to itself; 
the projection onto $T$ by $\Pi$ then gives us an infinite path satisfying both the LTL and the security constraint. 
Then among all such feasible combinations, we determine the optimal pair $(q_I^*,q_G^*)$ that minimizes the path cost function defined in~\eqref{eq2}  and the optimal plan $\tau=\tau^{q_I^*,q_G^*}[\tau^{q_G^*,q_G^*}]^\omega$ is returned.

\begin{algorithm}
	
	\SetAlgoNoLine 
	\SetKwInOut{Input}{\textbf{input}}\SetKwInOut{Output}{\textbf{output}}
	
	\Input{
		LTL formula $\phi$, WTS $T$ with $H$ and $Q_S$, initial state $q_0$
	}
	\Output{Optimal plan  $\tau$ from $q_0\in Q_0$
	}
	\BlankLine
	Convert $\phi$ to   NBA $B=(Q_B,Q_{0,B},\Sigma,\to_B,F_B)$\;
	Construct   twin-WTS  $V=(X,X_0,\to_{V}, w_V,\mathcal{AP},L_V)$\;
	Construct the product of $V$ and $B$ $T_{\otimes}=(Q_{\otimes},Q_{0,\otimes},\to_{\otimes},w_{\otimes})$\;
	\eIf{$\texttt{Reach}(\textsc{Int}_{q_0}(T_\otimes)) \cap  \textsc{Goal}(T_\otimes)=\emptyset$ }
	{
	  \textbf{return} ``no feasible plan from $q_0$"\;
	}
	{
	\For{$q_{I}\in \textsc{Int}_{q_0}(T_\otimes)$} 
	{
	\For{$q_{G}\!\in\! \texttt{Reach}(\{q_I\}) \!\cap \! \textsc{Goal}(T_\otimes)$} 
	{
	 $\tau^{q_I,q_G}=\Pi[\texttt{Shortpath}$($ q_{I}, q_{G}$)]\;
	 $\tau^{q_G,q_G}=\Pi[\texttt{Shortpath}$($q_{G}, q_{G}$)]\;
	}
    }
  	$(q_I^*,q_G^*)=\argmin_{(q_I,q_G)}\hat{J}(  \tau^{q_I,q_G} [ \tau^{q_G,q_G}] ^\omega  )$\; 
    \textbf{return} optimal plan $\tau=\tau^{q_I^*,q_G^*}[\tau^{q_G^*,q_G^*}]^\omega$ for $q_0$\;	
	}

	\caption{Security-Aware Optimal LTL Plan\label{al}}
\end{algorithm}

\begin{remark}
Let us discuss the complexity of Algorithm~\ref{al}. 
First, we note that the product system $T_\otimes$ contains at most  $|Q|^2|Q_B|$ states, where $|Q|$ is the number of states in the WTS model and $|Q_B|$ is the number of states in the B\"uchi automaton. 
Algorithm~\ref{al} involves at most $|Q|^4|Q_B|^2$ (very roughly estimated) shortest path problems which can be solved in polynomial-times in the number of states in $T_\otimes$. 
Therefore, the overall planning complexity is polynomial in both the number of states in the plant and the number of states in the B\"uchi automaton. 
Note that, in general, $|Q_B|$ is the length of $\phi$. 
However, in practice, the size of the LTL formula $\phi$ is usually very small and $Q$, which represents the state-space, is usually the main factor for scalability. 
\end{remark}

\subsection{Correctness of the Planning Algorithm}
Now, we prove the correctness of the proposed planning algorithm. 
Hereafter, we assume that the robot is starting from initial state $q_0$ and 
$\tau$ is the optimal plan from $q_0$ returned by  Algorithm~1. 
First, we show that the resulting plan satisfies the LTL task $\phi$.

\begin{mypro}\label{pro1}
	$\texttt{trace}(\tau) \models \phi$.
\end{mypro}
\begin{proof}
	We assume that optimal  path is obtained from the following projection
	$\tau= \Pi [ p^{pre} (p^{suf})^\omega ]$, where
	\begin{align}
		&p^{pre}=((q_0, q'_0), q_{0,B})\cdots ((q_n, q'_n), q_{n,B})\nonumber\\
		&p^{suf}=((q_{n+1}, q'_{n+1}), q_{n+1,B})\cdots ((q_{n+m}, q'_{n+m}), q_{n+m,B})\nonumber
	\end{align} 
and $q_{n+1,B}\in F_B$. 
That is, $\tau =q_0\cdots q_n(q_{n+1}\cdots q_{n+m})^\omega$.
According to the transition rule of $T_{\otimes}$, 
$\rho=q_{0,B}\cdots q_{n,B}(q_{n+1,B}\cdots q_{n+m,B})^\omega$ is an infinite run induced by 
infinite word $\texttt{trace}(\tau)=   L(q_0)\cdots L(q_n) (L(q_{n+1})\cdots L(q_{n+m}))^\omega$. 
Since $q_{n+1,B}\!\in\! F_B$, we know that $\texttt{Inf}(\rho)\cap F_B\neq \emptyset$, which means
$\texttt{trace}(\tau) \in \mathcal{L}_B=\texttt{Word}(\phi)$, i.e., $\texttt{trace}(\tau) \models \phi$.
\end{proof}

Second, we show that the planned path is secure.

\begin{mypro}\label{pro2}
	$\tau$ is secure.
\end{mypro}
\begin{proof}
Without loss of generality, we assume that $q_0\in Q_S$; otherwise, $\tau$ is secure trivially. 
Still, we assume that the optimal  path is obtained  by
$\tau= \Pi [  p^{pre} (p^{suf})^\omega ]$, where
\begin{align}
&p^{pre}=((q_0, q'_0), q_{0,B})\cdots ((q_n, q'_n), q_{n,B})\nonumber\\
&p^{suf}=((q_{n+1}, q'_{n+1}), q_{n+1,B})\cdots ((q_{n+m}, q'_{n+m}), q_{n+m,B}).\nonumber
\end{align} 
Then we know that 
$(q_0, q'_0)\cdots (q_n, q'_n) ( (q_{n+1}, q'_{n+1}) \cdots$ $ (q_{n+m}, q'_{n+m}) )^\omega \in \texttt{Path}^\omega(V)$.  
According to the transition rule of $V$, we have 
$H(\tau)= H(  q'_0\cdots q'_n (  q'_{n+1} \cdots q'_{n+m}  )^\omega  )$. 
Finally, since 
$((q_0, q'_0), q_{0,B}) \in \textsc{Int}_{q_0}(T_\otimes)$, we know that $q_0'\notin Q_S$. 
Therefore, $\tau'=q'_0\cdots q'_n (  q'_{n+1} \cdots q'_{n+m}  )^\omega $ is an internal path from a non-secret initial state having the same observation with $\tau$, i.e.,  $\tau$ is secure.
\end{proof}

Finally, we show that the planned path is optimal.

\begin{mypro}\label{pro3}
For any other secure path $\tilde{\tau}=\tilde{\tau}^{pre}[\tilde{\tau}^{suf}]^\omega$ such that $\texttt{trace}(\tilde{\tau}) \models \phi$, we have 
	$\hat{J}(\tau)\leq \hat{J}(\tilde{\tau})$.
\end{mypro}
\begin{proof}
We prove by contradiction.
Suppose that there exists a secure path $\tilde{\tau}\in \texttt{Path}^\omega(T)$ such that $\texttt{trace}(\tilde{\tau}) \models \phi$ and $\hat{J}(\tilde{\tau}) < \hat{J}(\tau)$. 
Since $\tilde{\tau}$ is secure, we know that there exists another path  $\tilde{\tau}'\in \texttt{Path}^\omega(T)$ such that $H(\tilde{\tau})=H(\tilde{\tau}')$ and $\tilde{\tau}'(1)\notin Q_S$. 
According to the definition of $V$, we know that there exists a path 
$\tau_V \in \texttt{Path}^\omega(V)$ in which the first component is  $\tilde{\tau}$ and the second component is $\tilde{\tau}'$. 
Furthermore, since  $\texttt{trace}(\tilde{\tau}) \models \phi$, by the definition of $T_\otimes$, 
there exists a path 
$\tau_\otimes \in \texttt{Path}^\omega(T_\otimes)$ in which the first component is  $\tau_V$ and the second component is $\tau_B$ such that $\texttt{Inf}(\tau_B)\cap F_B\neq \emptyset$.  Without loss of generality, we write 
$\tau_\otimes=p^{pre}(p^{suf})^{\omega}$ in the prefix-suffix structure, where 
\begin{align}
	&p^{pre}=((q_0, q'_0), q_{0,B})\cdots ((q_n, q'_n), q_{n,B})\nonumber\\
	&p^{suf}=((q_{n+1}, q'_{n+1}), q_{n+1,B})\cdots ((q_{n+m}, q'_{n+m}), q_{n+m,B}) \nonumber
\end{align}
and $q_{n+1,B}\in F_B$. 
Since $q'_0=\tilde{\tau}'(1)\notin Q_S$, we know that $\tilde{q}_I:=((q_0, q'_0), q_{0,B})\in \textsc{Int}_{q_0}(T_\otimes)$. Furthermore, we have 
$\tilde{q}_G:=((q_{n+1}, q'_{n+1}), q_{n+1,B}) \in  \texttt{Reach}(\{\tilde{q}_I\}) \!\cap \! \textsc{Goal}(T_\otimes)$.  
However, $\hat{J} (\tau^{\tilde{q}_,\tilde{q}_G}[\tau^{\tilde{q}_G,\tilde{q}_G}]^\omega )=  \hat{J}(\tilde{\tau}) < \hat{J}(\tau)$. 
This means that Algorithm~1 should at least output $\tau^{\tilde{q}_,\tilde{q}_G}[\tau^{\tilde{q}_G,\tilde{q}_G}]^\omega$ rather than $\tau$, which is a contradiction.
\end{proof}

The above three propositions show that the proposed algorithm is sound in the sense that the solution is correct if it finds one. Note that Algorithm~1 may return ``no feasible plan from $q_0$". Next we show that the proposed algorithm is also complete.

\begin{mypro}\label{pro4}
If  Algorithm~\ref{al} returns ``no feasible plan from $q_0$", then no solution to Problem 1 exists.
\end{mypro}
\begin{proof}
The proof is similar to the proof of Proposition~\ref{pro3}. 
Suppose, for the sake of contraposition, that there exists a secure path ${\tau}\in \texttt{Path}^\omega(T)$ such that $\texttt{trace}({\tau}) \models \phi$. 
Following the same argument in the proof of Proposition~\ref{pro3}, there exists a path 
$\tau_\otimes \in \texttt{Path}^\omega(T_\otimes)$, which is in the form of $\tau_\otimes=((\tau,\tau'),\tau_B)$ 
such that $\tau'(1)\notin Q_S$ and  $\texttt{Inf}(\tau_B)\cap F_B\neq \emptyset$. 
Therefore,  we have $\tilde{q}_I:=  ((\tau(1),\tau'(1)),\tau_B(1))   \in \textsc{Int}_{q_0}(T_\otimes)$. 
Furthermore,   $ \texttt{Reach}(\{\tilde{q}_I\}) \!\cap \! \textsc{Goal}(T_\otimes)\neq\emptyset$ since  $\texttt{Inf}(\tau_B)\cap F_B\neq \emptyset$. 
Therefore,  Algorithm~1 will not return ``no feasible plan from $q_0$".
\end{proof}

Finally, we summarize Propositions~\ref{pro1}, \ref{pro2}, \ref{pro3} and~\ref{pro4} by the following theorem.
\begin{mythm}\label{thm:1}
For any WTS $T=(Q,Q_0, \to, w,\mathcal{AP},L)$ with output function $H:Q\to Y$, secret states $Q_S$ and   LTL formula $\phi$, Algorithm~\ref{al} correctly solves the optimal security-aware LTL planning problem defined in Problem~\ref{prob1}.
\end{mythm}
\begin{proof}
The soundness of the algorithm is established by Propositions~\ref{pro1},~\ref{pro2} and~\ref{pro3}, 
and Proposition~\ref{pro4} shows the completeness of the algorithm.
\end{proof}

\section{Case Study}\label{sec:6}
\begin{figure}
	\center
	\includegraphics[width=0.34\textwidth]{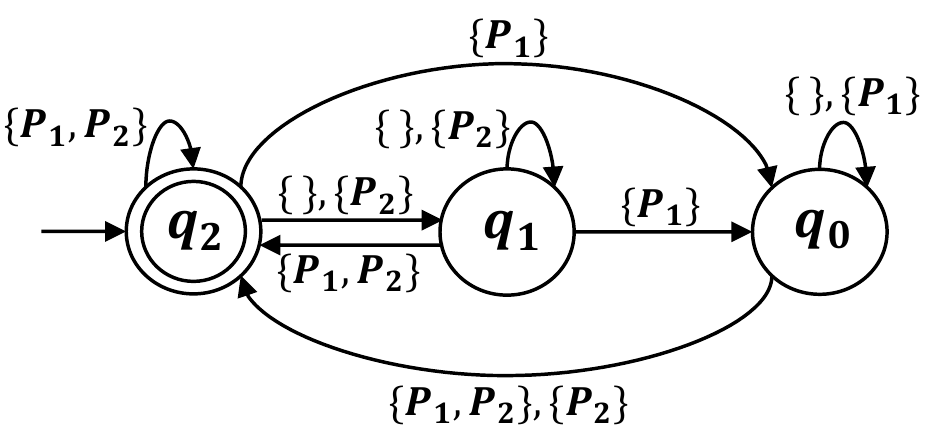}
	\caption{An NBA translated from $\phi=\square\lozenge P_1\wedge\square\lozenge P_2$.}\label{fig:nba}
\end{figure}

\begin{figure}
	\center
	\includegraphics[width=0.42\textwidth]{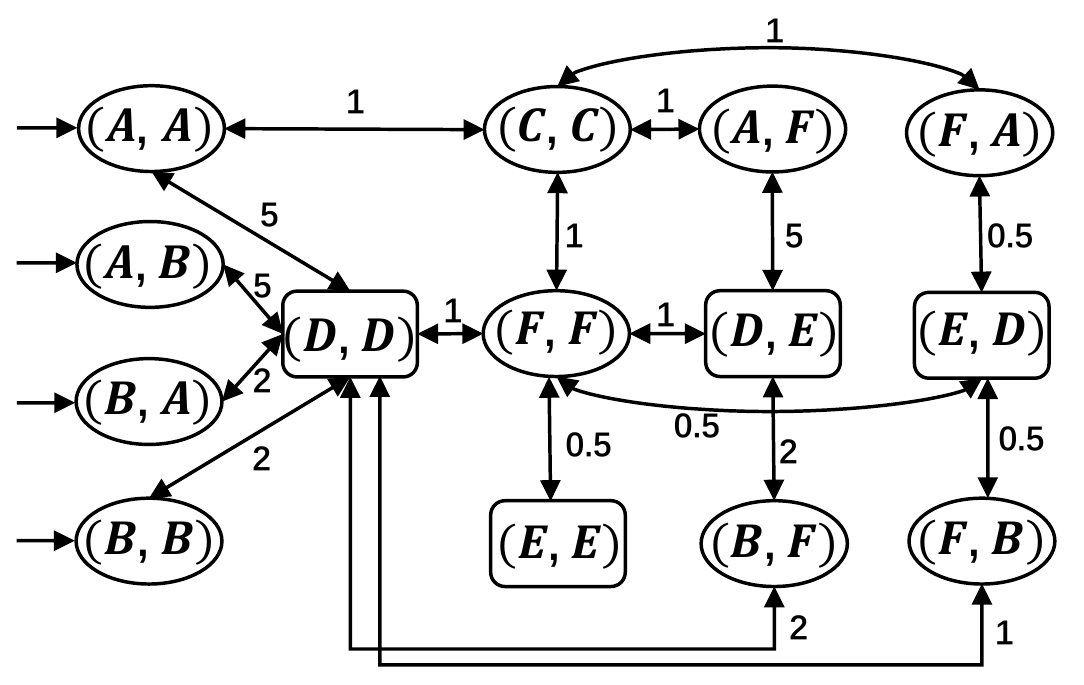}
	\caption{Twin-WTS $V$ of $T$ in Figure~\ref{fig:case}.}\label{fig:twin}
\end{figure}

\begin{figure}
	\center
	\includegraphics[width=0.46\textwidth]{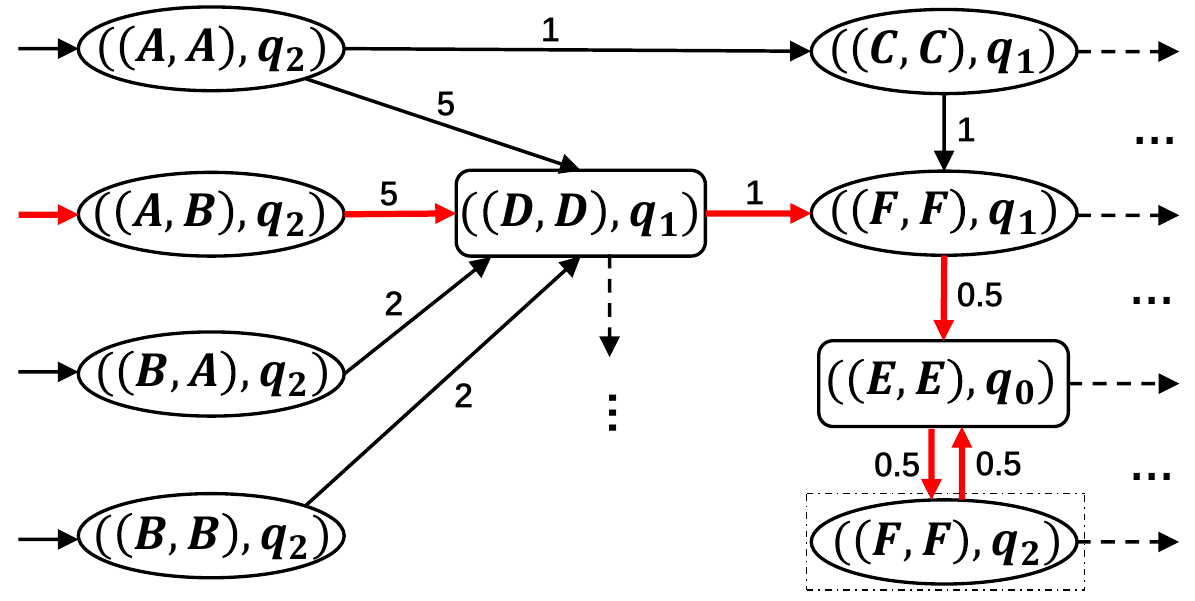}
	\caption{Example of the construction of the $T_{\otimes}$.  Red transitions represent the optimal feasible path. Due to limited space, some states and transitions are omitted and part of the product system is shown.}\label{fig:product}
\end{figure}
We go back to the motivating example in Section~\ref{sec:2} to illustrate the proposed planning algorithm. 
Consider again the WTS in Figure~\ref{fig:case}. 
To formalize the LTL task, we consider two atomic propositions $\mathcal{AP}=\{P_1, P_2\}$ with labeling function $L:Q\to 2^{\mathcal{AP}}$
defined by $L(F)=\{P_1\}$, $L(E)=\{P_2\}$ and $L(q)=\emptyset$ for other states. Then the task of the robot is expressed by the LTL formula
\[
\phi=\square\lozenge P_1\wedge \square\lozenge P_2.
\] 
The observation mapping is $H:Q\to \{Grass, Sand\}$ as specified in Figure~\ref{fig:case}. 
We define $Q_S=\{A\}\subseteq Q$, i.e., state $A$ is the unique  secret initial state.

To achieve the planning task, first we  convert $\phi$ to NBA $B=(Q_B,Q_{0,B},\Sigma,\to_B,F_B)$, which is shown in Figure~\ref{fig:nba}; such a conversion can be done by, e.g., the tool developed in \cite{gastin2001fast}. 
Then we construct the corresponding twin-WTS $V$, which is shown in Figure~\ref{fig:twin}. 
Specifically, $V$ contains four initial states $(A,A),(B,B), (A,B)$ and $(B,A)$ since $H(A)=H(B)=Sand$ and four combination are all  valid initial states.  
Then, for example, starting from $(A,B)$, only state $(D,D)$ can be reached as $A\to D, B\to D$ and $H(D)=H(D)=Grass$.  
Also, from state $(C,C)$, we can reach $(A,F)$ as  $C\to A, C\to F$ and $H(A)=H(F)=Sand$.
Finally, we need to construct the product system $T_{\otimes}$; for the sake of simplicity, we just show part of $T_{\otimes}$  in Figure~\ref{fig:product}, which is sufficient for the purpose of planning. 

Now, we assume that the robot is starting from  secret initial state $A$. 
Then we have  $\textsc{Int}_{A}(T_\otimes)=\{((A,B),q_2)\}$, which is a singleton.   
Also, we have $((F,F),q_2)\in \texttt{Reach}( \{((A,B),q_2)\}  ) \cap \textsc{Goal}(T_\otimes)$. 
One can check that such a state pair is indeed the one that minimizes the cost function if we draw the complete product system. 
Therefore, we obtain an optimal plan 
$\tau= \Pi[((A,B),q_2) ((D,D),q_1) ((F,F),q_1)  ( ((E,E),q_0)$ $ ((F,F),q_2)    )^\omega  ]     = AD(FE)^\omega$, 
which is highlighted by red transitions in Figure~\ref{fig:product}.

\section{Conclusion}\label{sec:7}
In this paper, we solved a  security-aware  optimal path planning problem for linear temporal logic tasks. 
A polynomial-time algorithm was proposed based on the product of the twin-system and the B\"uchi automaton.   
The synthesized solution is \emph{secure-by-construction} in the sense that it provides provably security guarantees for the designed systems against temporal logic tasks. 
Note that, in this work, we consider security requirement for  protecting the initial secret location of the system. 
In the future, we would like to extend the proposed algorithm to other types of security, e.g., infinite-step opacity. Also, we are interested in investigating optimal LTL path planning for multi-robot systems with security guarantees.
\bibliographystyle{ieeetr}
\bibliography{LTL}

\end{document}